\pdfoutput=1

\documentclass[sigconf]{aamas} 

\usepackage{balance} 

\usepackage{amsmath,amsfonts}
\usepackage{graphicx}
\usepackage{textcomp}
\usepackage{amsthm}
\usepackage{amsmath}
\usepackage{amsfonts}
\usepackage{booktabs}

\usepackage{url}

\newtheorem{theorem}{Theorem}[section]
\newtheorem{lemma}[theorem]{Lemma}
\newtheorem{prop}[theorem]{Proposition}
\theoremstyle{definition}
\newtheorem{definition}[theorem]{Definition}
\newtheorem{remark}[theorem]{Remark}
\newtheorem{observation}[theorem]{Observation}

\usepackage{times}
\usepackage{hyperref}       
\usepackage{nicefrac}       
\usepackage{microtype}      
\usepackage{doi}
\usepackage{MnSymbol}
\usepackage{wasysym}
\usepackage{pgfplots}
\pgfplotsset{compat=1.7}
\usepackage{tikz}
\RequirePackage{pgfplots}
\usepackage{algorithm}
\usepackage{algpseudocode}
\usepackage{algorithmicx}

\usepackage{comment}
\usepackage{mathtools}

\usepackage{dsfont}
\makeatletter
\newcommand\newtag[2]{#1\def\@currentlabel{#1}\label{#2}}
\makeatother
\usepackage{refcount}
\usepgfplotslibrary{external}
\usepackage[shortlabels]{enumitem}
\usepackage{titlesec}
\titleformat*{\section}{\Large\bfseries}
\titleformat*{\subsection}{\large\bfseries}
\titleformat*{\subsubsection}{\large\bfseries}
\newtheorem{innercustomthm}{Lemma}
\newenvironment{customthm}[1]
  {\renewcommand\theinnercustomthm{#1}\innercustomthm}
  {\endinnercustomthm}
\usepackage{apxproof}
\usepackage{longtable}

\definecolor{ForestGreen}{rgb}{.13,.54,.13}

\definecolor{Pink}{rgb}{.8,.4,.5}

\setcopyright{ifaamas}
\acmConference[AAMAS '23]{Proc.\@ of the 22nd International Conference
on Autonomous Agents and Multiagent Systems (AAMAS 2023)}{May 29 -- June 2, 2023}
{London, United Kingdom}{A.~Ricci, W.~Yeoh, N.~Agmon, B.~An (eds.)}
\copyrightyear{2023}
\acmYear{2023}
\acmDOI{}
\acmPrice{}
\acmISBN{}

\acmSubmissionID{892}

\title{Efficient Nearly-Fair Division with Capacity Constraints}

\author{Hila Shoshan}
\affiliation{
  \institution{Ariel University}
  \city{Ariel}
  \country{Israel}}
\email{hilashoshan0605@gmail.com}

\author{Noam Hazon}
\affiliation{
  \institution{Ariel University}
  \city{Ariel}
  \country{Israel}}
\email{noamh@ariel.ac.il}

\author{Erel Segal-Halevi}
\affiliation{
  \institution{Ariel University}
  \city{Ariel}
  \country{Israel}}
\email{erelsgl@gmail.com}

\begin{abstract}
We consider the problem of fairly and efficiently allocating  indivisible items (goods or bads) under capacity constraints.
In this setting, we are given a set of categorized items.
Each category has a capacity constraint (the same for all agents), that is an upper bound on the number of items an agent can receive from each category.
Our main result is a polynomial-time algorithm that solves the problem for two agents with additive utilities over the items.
When each category contains items that are all goods (positively evaluated) or all chores (negatively evaluated) for each of the agents,
our algorithm finds a feasible allocation of the items, which is both Pareto-optimal and envy-free up to one item.
In the general case, when each item can be a good or a chore arbitrarily, our algorithm finds an allocation that is Pareto-optimal and envy-free up to one good and one chore. 
\end{abstract}


\newcommand{\BibTeX}{\rm B\kern-.05em{\sc i\kern-.025em b}\kern-.08em\TeX}

\begin{document}


\pagestyle{fancy}
\fancyhead{}


\maketitle 



\section{Introduction}
The problem of how to fairly divide a set of items among agents with different preferences has been investigated by many mathematicians, economists, political scientists and computer scientists.
Most of the earlier work focused on how to fairly divide \emph{goods}, i.e., items with non-negative utility. 
In recent years, several works have considered the division of \emph{chores}, i.e., items with non-positive utility, and a few works also considered the division of a mixture of goods and chores (for example, \citet{aziz2022fair} and \citet{Brczi2020}).
Indeed, items may be  considered as goods for one agent and as chores for another agent. For example, consider a project that has to be completed by a team of students. It consists of several tasks that should be divided among the students, such as: programming tasks, user-interface tasks and algorithm development tasks. One student may evaluate the programming tasks as items with negative utilities and the UI and algorithmic tasks as items with positive utilities, while another student may evaluate them the other way around. 

Often, there is a constraint by which the items are partitioned into \emph{categories}, and each category has an associated \emph{capacity}, which defines the maximum number of items in this category that may be assigned to each agent. 
Considering again the student project example, 
the mentor of the project may want all students to be involved in all aspects of the project. Therefore, the mentor may partition the project tasks into three categories: programming, UI, and algorithms, setting a capacity for each category.
For example, if the team consists of two students, and there are $5$ programming tasks, $6$ UI tasks and $4$ algorithm tasks, then a capacity of $3$ on programming and UI tasks and a capacity of $2$ on algorithm tasks would ensure that both students are involved in about the same number of tasks from each category.
Clearly, the capacity constraints should be large enough so that all of the items in a given category could be assigned to the agents. An allocation satisfying all capacity constraints is called \emph{feasible}.

Note that, without capacity constraints, if one agent  evaluates an item as a good, while another agent evaluates it as a chore, we can simply give it to the agent who evaluates it as a good, as done by \citet{aziz2022fair}. However, with capacities it may not be possible, which shows that the combination of capacities and mixed valuations is more difficult than each of these on its own.

Two important considerations in item allocation are \emph{efficiency} and \emph{fairness}.
As an efficiency criterion, we use \emph{Pareto optimality} (PO), which means that no other feasible allocation is at least as good for all agents and strictly better for some agent.
As fairness criteria, we use two relaxations of \emph{envy-freeness} (EF). 
The stronger one is \emph{envy-freeness up to one item (EF1)}, which was
introduced by \citet{budish2011}, and 
adapted by \citet{aziz2022fair} for a mixture of goods and chores. Intuitively, an allocation is EF1 if for each pair of agents $i,j$, after removing the most difficult chore (for $i$) from $i$'s bundle, \emph{or} the most valuable good (for $i$) from $j$'s bundle, $i$ would not be jealous of $j$. 

With capacity constraints, an EF1 allocation may not exist. For example, consider a scenario with one category with two items, $o_1$ and $o_2$, and capacity constraint of $1$. $o_1$ is a good for both agents (e.g., $u_1(o_1)=u_2(o_1)=1)$, and $o_2$ is a chore for both agents (e.g., $u_1(o_2)=u_2(o_2)=-1)$. Clearly, in every feasible allocation, one agent must receive the good and the other agent must receive the chore (due to the capacity constraint), and thus the allocation is not EF1.
Therefore, we introduce a natural relaxation of it, which we call
\emph{envy-freeness up to one good and one chore (EF[1,1])}. It means that,
for each pair of agents $i,j$, there exists a chore in $i$'s bundle, and a good in $j$'s bundle, such that both are in the same category, and after removing them,
$i$ would not be jealous of $j$. 
In the special case in which, for each agent and category, either all items are goods or all items are chores (as in the student project example above), EF[1,1] is equivalent to EF1. We call this special case a \emph{same-sign instance}; note that it is still more general than only-goods or only-chores settings.

We focus on allocation problems between two agents. This case is practically important. For example, student projects are often done in teams of two, and household chores are often  carried out by the two partners. Fair allocation among two agents is the focus of various papers on fair division \cite{2.08.nicolo2008strategic,2.12.brams2012undercut,2.14.brams2014two,2.15.aziz2015note,2.17.nicolo2017divide,2.18.kilgour2018two,2.22.tucker2022playing,Brczi2020}.

We prove the existence of PO and EF[1,1] allocations with capacity constraints
for two agents with arbitrary (positive or negative) utilities over the items.
The proof is constructive:
we provide a polynomial-time algorithm that, for two agents, returns an allocation that is both PO and EF[1,1]. In a same-sign instance, the returned allocation is PO and EF1.

Our focus on the case of two agents allows us to simultaneously make two advancements over the state-of-the-art in capacity-constrained fair allocation \cite{biswas2018,dror2021fair}:
First, we handle a mixture of goods and chores, rather than just goods. As we show in Appendix \ref{app:dont-work},
standard techniques used for goods are not applicable for mixed utilities.
Second, we attain an allocation that is not only fair but also PO.
Before this work, it was not even known if a PO and EF1 allocation of goods with capacity constraints always exists.

Our algorithm  is based on the following ideas. 
The division problem can be considered as a matching problem on a bipartite graph, in which one side represents the agents and the other side represents the items. We add dummy items and clones of agents such that in every matching the capacity constraints are guaranteed. 
We assign a positive weight to each agent. We assign, to each edge between an agent and an item, a weight which is the product of the agent's weight and the valuation of the agent to the item.
A maximum-weight matching in this graph represents a feasible allocation that maximizes a weighted sum of utilities. 
Every allocation that maximizes a weighted sum of utilities, with positive agent weights, is Pareto-optimal.\footnote{
In fact, maximizing a weighted sum of utilities is stronger than Pareto-optimality. 
When allocating goods without capacity constraints, 
maximizing a weighted sum of utilities is equivalent to a stronger efficiency notion called \emph{fractional Pareto-optimality} \cite{negishi1960welfare,varian1976two,barman2018finding}.
}
Our algorithm first computes a maximum-weight matching that is also envy-free (EF) for one of the agents. It then tries to make it EF[1,1], while maintaining it a maximum-weight matching, by identifying pairs of items that can be exchanged between the agents, based on a ratio that captures how much one agent prefers an item relative to the other agent's preferences. 
Every exchange of items is equivalent to increasing the jealous agent's weight and decreasing the other agent's weight. 


\section{Related Work}
Fair division problems vary according to the nature of the objects being divided, the preferences of the agents, and the fairness criteria. Many algorithms have been developed to solve fair division problems, for details see the surveys of such algorithms \cite{brams1996fair},  \cite{moulin2004fair}, \cite{brams2007mathematics}, \cite{bouveret2016fair}.

In this paper we consider a new setting, which combines goods, chores, capacity constraints and Pareto-optimality. Note that even ignoring PO, goods, or both, our result is new.

\subsection{Mixtures of Goods and Chores}

\citet{Brczi2020} present a polynomial-time algorithm for finding an EF1 allocation for two agents with arbitrary utility functions (positive or negative).
\citet{Chen2020} proved that the leximin solution is EFX (a property stronger than EF1) for combinations of goods and chores, for agents with identical valuations.
\citet{gafni2021unified} 
present a generalization of both goods and chores, by considering items that may have several copies.
All these works do not consider efficiency.
Efficiency in a setting with goods and chores is studied by \citet{aziz2022fair}. They use the round-robin technique for finding an EF1 and PO division of combinations of goods and chores between two agents.
Similarly, \citet{aziz2020polynomial} find an allocation that is PROP1 (a property weaker than EF1) and PO for goods and chores.
\citet{aleksandrov2019greedy} 
prove that, with tertiary utilities, EFX and PO allocations always exist for mixed items.
However, all of these works do not handle capacity constraints.

\subsection{Constraints}
When all agents have weakly additive utilities, the round-robin protocol finds a complete EF1 division in which all agents receive approximately the same number of goods \cite{caragiannis2016}. 
This technique, together with the \emph{envy-graph}, has been used for finding a fair division of goods under capacity constraints \cite{biswas2018}.
This work has been extended to heterogeneous capacity constraints \cite{dror2021fair}, 
and to maximin-share fairness \cite{hummel2021guaranteeing}.

Fair allocation of goods of different categories has been studied by
\citet{mackin2016allocating}. The constraint is that each agent must receive at least one item per category. 
\citet{sikdar2019mechanism} consider an exchange market in which each agent holds multiple items of each category and should receive a bundle with exactly the same number of items of each category.
\citet{nyman2020fair} study a similar setting (they call the categories ``houses'' and the items ``rooms''), but with monetary transfers (``rent'').

Several other constraints have been considered. For example, 
\citet{bilo2022almost} study the fair division of goods such that each bundle needs to be connected on an underlying graph. \citet{igarashi2019} study PO allocation of goods with connectivity constraints.
An overview of the different types of constraints that have been considered can be found in \cite{suksompong2021constraints}.

\subsection{Efficiency and Fairness}
There are several techniques for finding a division of goods that is EF1 and PO. For example, the Maximum Nash Welfare algorithm selects a complete allocation that maximizes the product of utilities. It assumes that the agents' utilities are additive, and the resulting allocation is both EF1 and PO \cite{caragiannis2016,Wu2020}.

In the context of fair cake-cutting (fair division of a continuous resource), 
\citet{WELLER19855} proved the existence of an EF and PO allocation by considering the set of all allocations that maximize a weighted sum of utilities. We adapted this technique for the setting with indivisible items and capacity constraints.
\citet{barman2018finding} present a price-based mechanism that finds an EF1 and PO allocation of goods in pseudo-polynomial time.
Similarly, \citet{barman2019proximity} use a price-based approach to show that fair and efficient allocations can be computed in strongly polynomial time.
The price-based approach can be seen as a ``dual'' of our weight-based approach.

\citet{garg2022fair} present an algorithm for EF1 and PO allocation of chores when agents have bivalued preferences.
With general additive preferences,
the existence of an PO and EF1 allocation of chores for three agents (without capacity constraints) was proved only very recently by \citet{garg2022improving}. The authors claim that ``the case of chores turns out to be much more difficult to work with, resulting in relatively slow progress despite significant efforts by many researchers''. Indeed, for four or more agents, existence is still open even for only-chores instances and without capacity constraints.

\subsection{Alternative Techniques}
Our setting combines a mixture of goods and chores, capacity constraints, and a guarantee of both fairness and efficiency.
These three issues were studied in separation, but not all simultaneously.
Although previous works have developed useful techniques, they do not work for our setting.
For example, using the \emph{top-trading graph} presented by \citet{vaish2020} for dividing chores does not work when there are  capacity constraints.
The reason is that if we allocate an item to the ``\emph{sink}'' agent (i.e., an agent that does not envy any agent) on the top-trading graph, we may exceed the capacity constraints.
As another example, consider the maximum-weighted matching algorithm of \citet{brustle2019}. It is not hard to modify the algorithm to work with chores, but adding capacity constraints on each category might not maintain the EF1 property between the categories.
See Appendix \ref{app:dont-work} for more details.

Therefore, in this paper we develop a new technique for finding PO and EF1 (or EF[1,1]) allocation of the set of items (goods and chores), that also maintains capacity constraints.

Table~\ref{tbl:summary} summarizes some of the previous results mentioned in this section, which are close to our setting.

\begin{table*}[hbpt]
    \centering
\caption{Summary of some works on fair allocation of indivisible items}
\label{tbl:summary}
\begin{tabular*}
{\textwidth}{@{\extracolsep{\fill}}| p{7mm} | c | p{12mm} | c | c | p{40mm} | p{10mm} | c | p{38mm} |}
    \hline
    {paper} & {agents} & {utilities} & {goods} & {chores} & {constraints} & {fairness} & {PO} & {result} \\
    \midrule
    \cite{Brczi2020} & 2 & arbitrary & {v} & {v} & {-} & EF1 & {-} & {polynomial-time algorithm} \\
    \hline
    \cite{Chen2020} & any & identical & {v} & {v} & {-} & EFX & {-} & {the leximin solution} \\
    \hline
    \cite{gafni2021unified} & {any} & {leveled} & {v} & {v} & {-} & {EFX} & {-} & {existence proof} \\
    \hline
    \cite{aziz2022fair} & {2} & {arbitrary} & {v} & {v} & {-} & {EF1} & {v} & {round-robin technique} \\
    \hline
    \cite{aziz2020polynomial} & {any} & {arbitrary} & {v} & {v} & {-} & {PROP1} & {v} & {polynomial-time algorithm} \\
    \hline
    \cite{aleksandrov2019greedy} & {any} & {tertiary} & {v} & {v} & {-} & {EFX} & {v} & {existence proof} \\
    \hline
    \cite{caragiannis2016} & {any} & {weakly additive} & {v} & {-} & {approximately the same number} & {EF1} & {-} & {round-robin protocol} \\
    \hline
    \cite{biswas2018} & {any} & {additive} & {v} & {-} & {capacity constraints} & {EF1} & {-} & {round-robin protocol and envy-graph} \\
    \hline
    \cite{dror2021fair} & {any} & {heterog- eneous} & {v} & {-} & {heteroge- neous capacity constraint} & {EF1} & {-} & {polynomial-time algorithm} \\
    \hline
    \cite{hummel2021guaranteeing} & {any} & {additive} & {v} & {-} & {capacity constraint} & {MMS} & {-} & {polynomial-time algorithm} \\
    \hline
    \cite{mackin2016allocating} & {any} & {heterog- eneous and combinatorial} & {v} & {-} & {each agent gets at least one item per category} & {egalita- rian rank} & {-} & {characterize egalitarian + utilitarian rank-efficiency of categorial sequential allocation mechanisms.} \\
    \hline
    \cite{bilo2022almost} & {any} & {identical} & {v} & {-} & {each bundle needs to be connected on an underlying graph} & {EF1} & {-} & {polynomial-time algorithm} \\
    \hline
    \cite{igarashi2019} & {any} & {additive} & {v} & {-} & {bundles must be connected in an underlying item graph} & {EF1} & {v} & {non-existence on a path graph} \\
    \hline
    \cite{caragiannis2016} & {any} & {additive} & {v} & {-} & {-} & {EF1} & {v} & {max Nash welfare algorithm} \\
    \hline
    \cite{Wu2020}& {any} & {additive} & {v} & {-} & {each agent has a budget constraint on the total cost of items she receives} & {1/4-EF1} & {v} & {max Nash welfare algorithm} \\
    \hline
    \cite{barman2018finding} & {any} & {additive} & {v} & {-} & {-} & {EF1} & {v} & {pseudo-poly. time algorithm} \\
    \hline
    \cite{vaish2020} & {any} & {additive} & {-} & {v} & {-} & {EF1} & {-} & {polynomial-time algorithm} \\
    \hline
    \cite{brustle2019} & {any} & {additive} & {v} & {-} & {-} & {EF1} & {-} & {max weighted matching} \\
    \hline
    \cite{Bhaskar2020} & {3} & {additive} & {v} & {-} & {-} & {EFX} & {-} & {existence proof} \\
    \hline
    \cite{garg2022fair} & {any} & {additive, bivalued} & {-} & {v} & {-} & {EF1} & {v} & {polynomial-time algorithm} \\
    \hline
    \hline
    \textbf{We} & \textbf{2} & \textbf{additive} & \textbf{v} & \textbf{v} & \textbf{capacity constraints} & \textbf{EF1 || EF[1,1]} & \textbf{v} & \textbf{polynomial-time algorithm} \\
    \hline
\end{tabular*}
\end{table*}


\section{Notations}
An instance of our problem is a tuple $I = (N, M, C, S, U)$:
\begin{itemize}
    \item $N = [n]$ is a set of $n$ agents.
    \item $M = (o_1,\ldots,o_m)$ is a set of $m$ items.
    \item $C = (C_1, C_2, ..., C_k)$ is a set of $k$ categories. The categories are pairwise-disjoint and $M = \bigcup _{j} C_j$.
    \item $S = (s_1, s_2, ..., s_k)$ is a list of size $k$, containing the capacity constraint of each category. We assume that $\forall j \in [k]$: $\frac{|C_j|}{n} \leq s_j \leq |C_j|, s_j\in \mathds{N}$.
    The lower bound is needed to ensure we can divide all the items, and not "throw" anything away, and the upper bound is a trivial bound used for computing the run-time.
    \item
    $U$ is an $n$-tuple of utility functions $u_i : M \rightarrow \mathbb{R}$. We assume additive utilities, that is, $u_i(X) := \sum _{o \in X} u_i(o)$ for $X \subseteq M$.
\end{itemize}  

In a general \emph{mixed instance}, each utility can be any real number (positive, negative or zero).
A \emph{same-sign instance} is an instance in which, for each agent $i\in N$ and category  $j \in [k]$,
$C_j$ contains only goods for $i$ or only chores for $i$. That is, 
either $u_i(o)\geq 0$ for all $o\in C_j$, or 
$u_i(o)\leq 0$ for all $o\in C_j$.
Note that, even in a same-sign instance, it is possible that each agent evaluates different categories as goods or chores, and that different agents evaluate the same item differently.

An \emph{allocation} is a vector $A := (A_1, A_2, ..., A_n)$, 
with $\forall i,j \in [n], i \neq j : A_i \cap A_j = \emptyset$ and $\bigcup _{i \in [n]} A_i = M$.
$A_i$ is called "agent $i$'s bundle". 
An allocation $A$ is called \emph{feasible} if for all $i\in[n]$,
the bundle $A_i$ contains at most $s_c$ items of each category $C_c$, for each $c\in [k]$.

\begin{definition}[Due to \citet{aziz2022fair}]
\label{def:ef1}
An allocation $A$ is called \textit{Envy Free up to one item (EF1)} if for all $i,j\in N$, at least one of the following holds:
\begin{itemize}
\item $\exists T\subseteq A_i$ with $|T|\leq 1$,
s.t. $u_i(A_i\setminus T) \geq u_i(A_j)$.
\item $\exists G\subseteq A_j$ with $|G|\leq 1$,
s.t. $u_i(A_i) \geq u_i(A_j\setminus G)$.
\end{itemize}
\end{definition}

We also define a slightly weaker fairness notion, that we need for handling general mixed instances, in which an EF1 allocation is not guaranteed to exist, as shown in Introduction.

\begin{definition}
\label{def:ef11}
An allocation $A$ is called \textit{Envy Free up to one good and one chore (EF[1,1])} if for all $i,j\in N$, there exists a set $T\subseteq A_i$ with $|T|\leq 1$,
and a set $G\subseteq A_j$ with $|G|\leq 1$,
such that G and T are of the same category,
and 
$u_i(A_i \diagdown T) \geq u_i(A_j \diagdown G)$.
\end{definition}

The uncategorized setting of \citet{aziz2022fair} can be reduced to our setting by putting each item in its own category, with a capacity of $1$.
An allocation is EF[1,1] in the categorized instance if-and-only-if it is EF1 (by Definition \ref{def:ef1}) in the original instance.

Throughout the paper, any result that is valid for mixed instances with EF[1,1] is also valid for same-sign instances with EF1. 
This follows from the following lemma. 
\begin{lemma}
\label{EF1equivalent}
In a same-sign instance, EF[1,1] is equivalent to EF1.
\end{lemma}
\begin{proof}
Suppose that some allocation, $A$, for a same-sign instance is EF[1,1]. Therefore, for all $i,j\in N$, $\exists T\subseteq A_i$ with $|T|\leq 1$,
and $\exists G\subseteq A_j$ with $|G|\leq 1$, such that G and T are of the same category, and
$u_i(A_i \diagdown T) \geq u_i(A_j \diagdown G)$.

If $|G|=0$ or $|T|=0$, then $A$ is EF1, by definition. So assume that $|G|=|T|=1$. 
Since G and T are in the same category, and in a same-sign instance, for each agent $i\in [n]$ and category $c\in [k]$, $C_c$ contains only goods for $i$ or only chores for $i$, then, for all $j\in [n]$, 

if $C_c$ is a category of goods for agent $i$, then $u_i(A_i \diagdown T) \geq u_i(A_j \diagdown G)$ implies 
$u_i(A_i)\geq u_i(A_j \diagdown G)$, 
so both allocations are EF1 for agent $i$.
If $C_c$ is a category of chores for agent $i$, then $u_i(A_i \diagdown T) \geq u_i(A_j \diagdown G)$ implies 
$u_i(A_i \diagdown T)\geq u_i(A_j)$,
so again both allocations are EF1 for agent $i$.
\end{proof}

\begin{remark}
Our new EF[1,1] is reminiscent of another guarantee called $EF_1^1$, that is, envy-freeness up to adding a good to one agent and removing a good from another agent \cite{barman2019proximity}. But lemma \ref{EF1equivalent} implies that EF[1,1] is stronger. The reason is that if there are only goods, it is enough to remove one good from an agent's bundle, and there is no need to also add a good to the envious agent's bundle.

EF[1,1] can be seen as a generalization of EF1 as defined in [Aziz et al. 2022] to the case of categorized items (you just have to define one category for every item, with an upper bound equal to one).
\end{remark}

\begin{remark}
The restriction in Definition \ref{def:ef11} that $G$ and $T$ should be of the same category is essential for Lemma \ref{EF1equivalent}. 
To see this, denote by EF[1,1,U] the unrestricted variant of EF[1,1], allowing to remove one chore and one good from any category.
Suppose that there are two categories: one of them contains a good (for both agents) and the other contains a chore (for both agents). If one agent gets the good and the other agent gets the chore, the allocation is EF[1,1,U], and it is a same-sign instance, but it is not EF1.

Any EF[1,1] allocation is clearly EF[1,1,U]. Therefore, proving that our algorithm returns an EF[1,1] allocation implies two things at once: in general instances, it returns an EF[1,1,U] allocation;
and in same-sign instances, our algorithm returns an EF1 allocation.
\end{remark}

Finally, we recall two definitions:
\begin{definition}
Given an allocation $A$ for $n$ agents, the \emph{envy graph} of $A$ is a graph with $n$ nodes, each represents an agent, and there is a directed edge $i\rightarrow j$ iff $i$ envies $j$ in allocation $A$.
A cycle in the envy graph is called an \emph{envy cycle}. 
\end{definition}

Our efficiency criterion is defined next:

\begin{definition}
Given an allocation $A$, another allocation $A'$ is a \textit{Pareto-improvement} of $A$ if $u_i(A'_i) \geq u_i(A_i)$ for all $i \in N$, and $u_j(A'_j) > u_j(A_j)$ for some $j \in N$.

A feasible allocation $A$ is \textit{Pareto-Optimal} $(PO)$ if no feasible allocation is a Pareto-improvement of $A$.
\end{definition}

\section{Finding a PO and EF[1,1] Division}
In this section, we present some general notions that can be used for any number of agents.

Then, we present our algorithm that finds in polynomial time a feasible PO allocation with two agents.
In any mixed instance, this allocation is also EF[1,1]; in a same-sign instance, it is also EF1, according to Lemma \ref{EF1equivalent}.

\subsection{Preprocessing}
We preprocess the instance such that,
in any feasible allocation, all bundles have the same cardinality.
To achieve this, we add to each category 
$C_c$ with capacity constraint $s_c$,
some $ns_c - |C_c|$ dummy items with a value of $0$ to all agents.
In the new instance, each bundle must contain exactly $s_c$ items from each category $C_c$.
From now on, without loss of generality, we assume that $|M|=m=\sum_{c\in[k]} ns_c$.
This implies that, in every feasible allocation $A$, we have $|A_i| = m/n$ for all $i\in [n]$.

\subsection{Maximizing a Weighted Sum of Utilities}
Our algorithm is based on searching the space of PO allocations. 
Particularly, we consider allocations that maximize a weighted sum of utilities $w_1 u_1 + w_2 u_2 + ... + w_n u_n$, where each agent $i$ is associated with a weight $w_i \in [0,1]$, and $w_1+w_2+...+w_n=1$. Such allocations can be found by solving a maximum-weight matching problem in a weighted bipartite graph. 
We denote the set of all agents' weights by $w=(w_1,w_2,...,w_n)$.

\begin{definition}
For any $n$ real numbers (weights)
$w=(w_1,w_2,...,w_n)$, such that,
$\forall i \in [n], w_i \in [0,1]$, 
and $w_1+w_2+...+w_n=1$,
let 
$G_w$ be a bipartite graph $(V_1 \cup V_2,E)$ with  $|V_1|=|V_2|=m$. 
$V_2$ contains all $m$ items (of all categories, including dummies). $V_1$ contains $\frac{m}{n}$ copies of each agent $i\in [n]$.
For each category $c\in [k]$, we choose distinct $s_c$ copies of each agent and add an undirected edge from each of them to all the $n s_c$ items of $C_c$.
Each edge $\{i,o\} \in E$, $i\in V_1, o\in V_2$ has a weight $w(i,o)$, where: 
\[   
w(i,o) := w_i\cdot u_i(o)
\]

An allocation is called \emph{$w$-maximal} if it corresponds to a maximum-weight matching among the maximum-cardinality matchings in $G_w$.
\end{definition}

\begin{prop}
\label{all-PO}
Every $w$-maximal allocation, where $w_1,w_2,\ldots\allowbreak,w_n\in (0,1)$, is PO.
\end{prop}
\begin{proof}
Every $w$-maximal allocation $A=(A_1,A_2,...,A_n)$ maximizes the sum $w_1u_1(A_1)+w_2u_2(A_2)+...+w_nu_n(A_n)$.
Every Pareto-improvement would increase this sum. Therefore, there can be no Pareto-improvement, so $A$ is PO.
\end{proof}

\subsection{Exchanging Pairs of Items}
Our algorithm starts with a $w$-maximal allocation, and repeatedly exchanges pairs of items between the agents in order to find an allocation that is also EF[1,1]. To determine which pairs to exchange, we need some definitions and lemmas.

\begin{definition}
Given a feasible allocation $A = (A_1,A_2,...,A_n)$, an \emph{exchangeable pair} is a pair $(o_i,o_j)$ of items, $o_i\in A_i$ and $o_j\in A_j$, $i,j\in [n], i\neq j$, 
such that 
$o_i$ and $o_j$ are in the same category.

This implies that $A_i\setminus \{o_i\}\cup \{o_j\}$
and $A_j\setminus \{o_j\}\cup \{o_i\}$ are both feasible. Additionally, in a same-sign instance, 
for each agent, $o_i, o_j$ are in the same "type", 
that is, both goods or both chores.
\end{definition}

In this paper, we work a lot with exchangeable pairs, so we use
$o_i,o_j \in A_i,A_j$ as a shorthand for ``$o_i\in A_i$ and $o_j\in A_j$''.

\subsubsection{Finding a Fair Allocation}

The following two lemmas deal with fairness while exchanging exchangeable pairs in a $w$-maximal allocation.

\begin{lemma}
\label{both-EF1-for-someone}
Let $A$ be a $w$-maximal feasible allocation, 
and let $A'$ be another feasible allocation, resulting from $A$ by exchanging an exchangeable pair $(o_{i},o_{j})$ between some two agents $i\neq j$. 
Then there exists some ordering of the agents, $k_1,\ldots,k_n$, such that for all $y>x$, the EF[1,1] condition is satisfied for agent $k_y$ with respect to agent $k_x$ in \emph{both} allocations $A$ and $A'$. 
That is, $k_y$ envies $k_x$ up to one good and one chore in both allocations. 

In particular, there is at least one agent (agent $k_n$) for whom \emph{both} $A$ and $A'$ are EF[1,1].
\end{lemma}
\begin{proof}
Let $A=(A_1,..,A_n)$ and
$A' = (A_1',...,A_n')$.
Let $C_c$ be the category that contains both items $o_i, o_j$.
By the pre-processing step, every bundle in $A$ contains at least one item from $C_c$. So 
we can write every bundle $A_x$, for all $x\in [n]$, as: $A_x = B_x\cup \{o_x\}$ for some $o_x\in C_c$.
After the exchange, we have for all 
$x\neq i,j: A_x' = A_x = B_x\cup \{o_x\}$, whereas $A_i' = B_i\cup \{o_j\}, A_j' = B_j\cup \{o_i\}$.  

Consider the envy-graph representing the partial allocation $(B_1,B_2,\allowbreak\ldots,B_n)$.
We claim that it contains no cycle. Suppose that it contained an envy-cycle. If we replaced the bundles according to the direction of edges in the cycle, we would get another feasible allocation which is a Pareto-improvement of the current allocation, $A$, which is $w$-maximal. Contradiction!

Therefore, the envy-graph of 
$(B_1,B_2,...,B_n)$ has a topological ordering. 
Let $k_1,\ldots,k_n$ be such an ordering,
so that for all $y>x$, agent $k_y$ prefers $B_{k_y}$ over $B_{k_x}$.
In both allocations $A$ and $A'$, the bundles of both $k_y$ and $k_x$ are derived from $B_{k_y}$ and $B_{k_x}$ by adding a single good or chore. Therefore, in both $A$ and $A'$, the EF[1,1] condition is satisfied for agent $k_y$ w.r.t. agent $k_x$.
In particular, for agent $k_n$, both these allocations are EF[1,1].\footnote{In fact, the result holds not only for an exchange of two items, but also for any permutation of $n$ items of the same category, one item per agent. The proof is the same.}
\end{proof}
Lemma \ref{both-EF1-for-someone} considered a single exchange. Now, we consider a sequence of exchanges. 
The following lemma works only for two agents --- we could not yet extend it to more than two agents.
\begin{lemma}
\label{exists-PO-and-EF1}
Suppose there are $n=2$ agents.
Suppose there is a sequence of feasible allocations $A^1,\ldots,A^x$ with the following properties:
\begin{itemize}
\item $\forall j\in[x]$, the allocation $A^j=(A_1^j,A_2^j)$ is $w$-maximal, where $w=(w_{1,j},w_{2,j})$ for some $w_{1,j},w_{2,j}\in(0,1)$.
\item $A^1$ is EF for agent 1 and $A^x$ is EF for agent 2.
\item $\forall j\in[x-1]$, $A^{j+1}$ is obtained from $A^{j}$ by a single exchange of an exchangeable pair between the agents. 
\end{itemize}
Then, for some $j\in[x]$, the allocation $A^j$ is PO and EF[1,1].
\end{lemma}

\begin{proof}
Every $A^j$ is PO by Proposition \ref{all-PO}. 
Therefore, it is never possible for the two agents to envy each other simultaneously.
Since at $A^1$ agent 1 is not jealous and at $A^x$ agent 2 is not jealous, there must be some $j\in [x-1]$ in which $A^{j}$ is EF for 1, and $A^{j+1}$ is EF for 2. 

Because $A^{j+1}$ results from $A^j$ by exchanging an exchangeable pair between the agents, by Lemma \ref{both-EF1-for-someone}, there exists an agent $i \in [2]$ such that both $A^{j}$ and $A^{j+1}$ are EF[1,1] for $i$.

If both are EF[1,1] for agent 1, then $A^{j+1}$ is an EF[1,1] allocation. If both are EF[1,1] for agent 2, then $A^j$ is an EF[1,1] allocation.
\end{proof}

To apply Lemma \ref{exists-PO-and-EF1},
we need a way to choose the pair of exchangeable items in each step of the sequence, so that the next allocation in the sequence remains $w$-maximal.
We use the following definition.

\begin{definition}
\label{diff_ratio_definition}
For a pair of agents $i,j\in [n]$ s.t. $i\neq j$, and a pair of items $(o_i,o_j)$, the \emph{difference ratio}, denoted by $r_{j/i}(o_i,o_j)$, is defined as:
\[
r_{j/i}(o_i,o_j) := \frac{u_j(o_i)-u_j(o_j)}{u_i(o_i)-u_i(o_j)}
\]
If $u_j(o_i)= u_j(o_j)$, then the ratio is always 0.
If $u_i(o_i) = u_i(o_j)$ (and $u_j(o_i) \neq u_j(o_j)$), then the ratio is defined as $+\infty$
if $u_j(o_i)> u_j(o_j)$, or $-\infty$ if $u_j(o_i) < u_j(o_j)$.
\end{definition} 

\subsubsection{The Properties of a $w$-maximal Allocation}

The following lemma is proved in Appendix \ref{app:lem:diffs}.

\begin{lemma}
\label{lem:diffs}
For any $n$ agents,
for any $w=(w_1,w_2,...,w_n)$ such that $w_1,w_2,...,w_n\in(0,1)$, and an allocation $A=(A_1,...,A_n)$, the following are equivalent:

(i) $A$ is $w$-maximal.


(ii) Every exchange-cycle does not increase the weighted sum of utilities.
That is, for all $x\geq 2$, a subset of agents $\{a_1,...,a_x\}\in [n]$, and a set of items $o_1,...,o_x$, such that all are in the same category, and $\forall j\in [x], o_j\in A_{a_j}$:
\begin{align*}
    w_{a_1} u_{a_1} (o_1) + w_{a_2} u_{a_2} (o_2) + ... + w_{a_x} u_{a_x} (o_x) \geq \\
    w_{a_1} u_{a_1} (o_x) + w_{a_2} u_{a_2} (o_1) + ... + w_{a_x} u_{a_x} (o_{x-1})
\end{align*}
\end{lemma}

The following lemma follows from Lemma \ref{lem:diffs}, but only for two agents.

\begin{lemma}
\label{lem:diffs2}
Suppose there are $n=2$ agents.
For any $w_1,w_2\in(0,1)$ and an allocation $A=(A_1,A_2)$, the following are equivalent:

(i) $A$ is $w$-maximal, for $w=(w_1,w_2)$.

(ii) For any exchangeable pair $o_1,o_2 \in A_1,A_2$, exactly one of the following holds:
\begin{align*}
u_1(o_1) > u_1(o_2) && \text{and} && 
w_1/w_2 \geq r_{2/1}(o_1,o_2) && \text{or} 
\\
u_1(o_1) = u_1(o_2) && \text{and} && 
u_2(o_2)\geq u_2(o_1) && \text{or} 
\\
u_1(o_1) < u_1(o_2) && \text{and} && 
w_1/w_2 \leq r_{2/1}(o_1,o_2)
\end{align*}
\end{lemma}

\begin{proof}
The only exchange-cycle in a 2-agents allocation is a replacement of an exchangeable pair $o_1,o_2 \in A_1,A_2$ between the agents.
Then, according to Lemma \ref{lem:diffs}, for any exchangeable pair $o_1, o_2\in A_1,A_2$,
\begin{align}
\label{eq:diff0}
w_1 u_1(o_1) + w_2 u_2 (o_2)
\geq
w_1 u_1(o_2) + w_2 u_2 (o_1)
\\
\label{eq:diff1}
w_1 u_1(o_1) - w_2 u_2 (o_1)
\geq
w_1 u_1(o_2) - w_2 u_2 (o_2)
\\
\label{eq:diff2}
w_1 [u_1(o_1)-u_1(o_2)]
\geq 
w_2 [u_2(o_1)-u_2(o_2)]
\end{align}
The claim in (ii) is an algebraic manipulation of \eqref{eq:diff2}, so (ii) $\iff$ \eqref{eq:diff2}. And since (i) $\iff$ \eqref{eq:diff2}, also (i) $\iff$ (ii).
\end{proof}

\begin{lemma}
\label{cases-lemma}
For any $n$ agents,
in any $w$-maximal allocation $A$ (with positive weights), for any $i,j$ and an exchangeable pair $o_i, o_j\in A_i, A_j$, the following implications hold:
\begin{align*}
u_j(o_i) \geq u_j(o_j)
&& \implies && 
u_i(o_i) \geq u_i(o_j)
\\
u_j(o_i) > u_j(o_j)
&& \implies && 
u_i(o_i) > u_i(o_j)
\end{align*}
\end{lemma}

\begin{proof}
By Lemma \ref{lem:diffs}, since $A$ is a $w$-maximal allocation, each exchange-cycle does not increase the sum of the matching. In particular, for $x=2$, if we define $a_1=i, a_2=j, o_1=o_i, o_2=o_j$, we have: 
\begin{align*}
    w_i u_i (o_i) + w_j u_j (o_j)
    \geq 
    w_i u_i (o_j) + w_j u_j (o_i)
\end{align*}
Which is equal to: 
\begin{align*}
    w_i [u_i (o_i) - u_i (o_j)]
    \geq 
    w_j [u_j (o_i) - u_j (o_j)]
\end{align*}
$w_i$ and $w_j$ are both positive, so if the left term is positive or non-negative, the right term must to be positive or non-negative too, respectively.  
\end{proof}

Lemma \ref{cases-lemma} implies that, in any exchangeable pair $o_i,o_j \in A_i, A_j$ in a $w$-maximal allocation, there are two cases:
(a) Both agents prefer the same item ($o_i$ or $o_j$);
(b) Agent $i$ prefers $o_i$ and agent $j$ prefers $o_j$.
In case (a), we say that the exchangeable pair has a \emph{preferred item}.

\begin{definition}
\label{def:preffered}
Consider a $w$-maximal allocation $A$ and an exchangeable pair  $o_i, o_j\in A_i, A_j$, for some $i,j\in [n]$. 
$o_i$ is called a \textit{preferred item} in the exchangeable pair $(o_i,o_j)$ if both $u_j(o_i)>u_j(o_j)$ and $u_i(o_i)>u_i(o_j)$.
\end{definition}

\begin{lemma}
\label{lem:envy}
For any $n$ agents,
in any $w$-maximal allocation $A$,
if an agent $j$ envies some agent $i$,
then there is an exchangeable pair $o_i, o_j\in A_i, A_j$, 
and $o_i$ is the preferred item.
\end{lemma}
\begin{proof}
If $j$ envies $i$,
then $u_j(A_i)>u_j(A_j)$.
Since both $A_i$ and $A_j$ contain the same number of items in each category, there must be a category in which, for some item pair $o_i, o_j\in A_i, A_j$,
agent $j$ prefers $o_i$ to $o_j$.
By Lemma \ref{cases-lemma}, agent $i$ too prefers $o_i$ to $o_j$. So $o_i$ is a preferred item.
\end{proof}

\subsubsection{Maintaining the $w$-maximality}
The following lemma
shows that, by exchanging items, we can move from one $w$-maximal allocation to another $w'$-maximal allocation (for a possibly different weight-vector $w'$).
This lemma, too, works only for two agents.
\begin{lemma}
\label{lem:exchange}
Suppose there are $n=2$ agents.
Let $A$ be a $w$-maximal allocation, for $w=(w_1,w_2)$. Suppose there is an exchangeable pair $o_1, o_2\in A_1, A_2$ such that:
\begin{enumerate}
\item $u_2(o_1)>u_2(o_2)$, that is, $o_1$ is the preferred item.
\item \label{max-cond} Among all exchangeable pairs in which $o_1$ is the preferred item, this pair has a largest difference-ratio $r_{2/1}(o_1,o_2)$.
\end{enumerate}
Let $A'$ be the allocation resulting from exchanging $o_1$ and $o_2$ in $A$.
Then, $A'$ is $w'$-maximal for some $w' = (w_1',w_2')$ with  $w_1'\leq w_1, w_2'\geq w_2, w_1'\in(0,1), w_2'\in(0,1)$.
\end{lemma}
\begin{proofsketch}
The lemma can be proved by using Lemmas \ref{lem:diffs2}, \ref{cases-lemma}, the maximality condition in the lemma [condition \ref{max-cond}] and Definition \ref{diff_ratio_definition}.

The idea of the proof is to define $w_1',w_2'\in (0,1)$ such that $\frac{w_1'}{w_2'} = r_{2/1}(o_1,o_2), w_1'+w_2'=1$. Then, $0 < \frac{w_1'}{w_2'} \leq \frac{w_1}{w_2}$, and $w_1'\leq w_1, w_2'\geq w_2$.

Then we look at all the exchangeable pairs $(o_1^*,o_2^*)$ in the new allocation $A'$, resulting from the exchange, and show that they satisfy all the conditions of Lemma \ref{lem:diffs2}(ii) with $w_1',w_2'$, which are:
\begin{enumerate}[(a)] 
\item $u_1(o_1^*) > u_1(o_2^*)$ and 
    $r_{2/1}(o_1,o_2) \geq r_{2/1}(o_1^*,o_2^*)$ or
\item $u_1(o_1^*) = u_1(o_2^*)$ and
    $u_2(o_2^*)\geq u_2(o_1^*)$ or
\item $u_1(o_1^*) < u_1(o_2^*)$ and
    $r_{2/1}(o_1,o_2) \leq r_{2/1}(o_1^*,o_2^*)$
\end{enumerate}

The exchangeable pairs in $A'$ can be divided into four types:
\begin{enumerate}
    \item The exchangeable pairs $(o_1^*,o_2^*)$ that have not moved. 
    \item The pair $(o_2,o_1)$.
    \item Pairs in the form $(o_1^*,o_1)$, $o_1^*\in A_1', o_1^* \neq o_2$.
    \item Pairs in the form $(o_2,o_2^*)$, $o_2^*\in A_2', o_2^* \neq o_1$.
\end{enumerate}

We show that each pair of each type satisfies its own condition out of (a), (b) and (c).
Therefore, by Lemma \ref{lem:diffs2}, $A'$ is $w'$-maximal allocation, for $(w_1',w_2')$.

The complete proof with all the technical arguments can be found in Appendix \ref{app:lem:exchange}.
\end{proofsketch}

\subsection{Algorithm for Two Agents}

Throughout this subsection we consider general mixed instances, for simplicity. 
By Lemma \ref{EF1equivalent}, for same-sign instances all the results hold with EF1 instead of EF[1,1]. 

Let us start with an intuitive description of the algorithm, for two agents.
Suppose that $w_2$ is a function of $w_1$, and consider the line $w_1+w_2=1, w_1\geq 0, w_2\geq 0$, which describes the collection of all pairs of non-negative weights $w_1, w_2\in [0,1]$ whose sum is 1.
Each point on this line represents a $w'$-maximal allocation, for some weight-vector $w'$.
In every such allocation, there are no envy-cycles in the envy graph, so there is at most one envious agent. 

The algorithm starts with an initial allocation which is a maximum-weight matching in the graph $G_w$, where $w=(0.5,0.5)$, corresponding to the center of the line.
This initial allocation is PO (By Lemma \ref{all-PO}) and EF for at least one agent. If it is EF for both agents then we are done. Otherwise, 
depending on the envious agent, the algorithm decides which side of the line to go to. If agent 2 envies, we need to improve 2's weight, so we go towards (0,1). If agent 1 envies, we need to go towards (1,0).
Therefore, as long as the allocation is not EF[1,1], the algorithm swaps an exchangeable pair chosen according to Lemma \ref{lem:exchange}, thus maintaining the search space as the space of the $w$-maximal allocations.
Note that since the items of the exchanged pair are both in the same category, the capacity constraints are also maintained. 
Lemma \ref{exists-PO-and-EF1} implies that
some point on the line gives a feasible EF[1,1] and PO division.

Specifically, the exchange pairs are determined as follows.
For each item $o$ we can define a linear function $f_o(w_1)$:
\begin{align*}
    w_1 u_1(o) - w_2 u_2(o) 
    =&
    w_1 u_1(o) - (1-w_1) u_2(o) \\
    =&
    w_1 u_1(o) - u_2(o) + w_1 u_2(o) \\
    =& 
    (u_1(o) + u_2(o)) w_1 - u_2(o)
\end{align*}

If we draw all those functions in one coordinate system, each pair of lines intersects at most once. 
In total there are $O(m^2)$ intersections, where $m= \sum_{c\in [k]}|C_c|$, the total number of items, in all categories (including the dummies).

For example, consider the same-sign instance $I=(N, M, C, S, U)$ where $N=[2], C=\{C_1,C_2\}$, $C_1=\{o_1,o_2,o_3,o_4\}, C_2=\{o_5,o_6\}$, $S=\{2,1\}$ and $U$ is shown in Table~\ref{tbl:util_example}. 
\begin{table}[t]
\caption{Utilities of the agents in the example.}
\label{tbl:util_example}
\begin{tabular}{||c|c c  c c |c c||} 
    \hline
    {} & {$o_1$} & {$o_2$} & {$o_3$}& {$o_4$}& {$o_5$} & {$o_6$}\\
\midrule
Agent 1 & \newtag{0}{u_1(o_11)} 
            & \newtag{-1}{u_1(o_12)} 
            & \newtag{-4}{u_1(o_13)} 
            & \newtag{-5}{u_1(o_14)} 
            & \newtag{0}{u_1(o_21)} 
            & \newtag{2}{u_1(o_22)} \\
    \hline
Agent 2 & \newtag{0}{u_2(o_11)} 
        & \newtag{-1}{u_2(o_12)}
        & \newtag{-2}{u_2(o_13)}
        & \newtag{-1}{u_2(o_14)}
        & \newtag{-1}{u_2(o_21)}
        & \newtag{0}{u_2(o_22)} \\ [1ex] 
    \hline
\end{tabular} 
\end{table}
The corresponding lines for the items are depicted in Figure~\ref{fig:lines}.
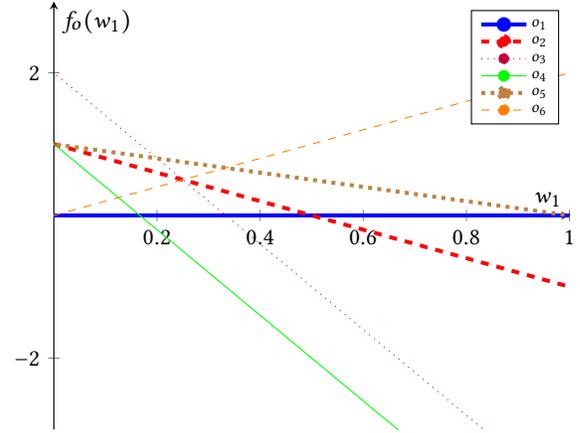
\begin{figure}[tb]
\begin{center}
\begin{tikzpicture}
\begin{axis}[xlabel=$w_1$,ylabel=$f_o(w_1)$,
xmin=0,xmax=1,ymin=-3,ymax=3, axis lines=center,
legend style={nodes={scale=0.7, transform shape}}, legend image post style={mark=*}]

\addplot[domain=-7:7, color=blue, line width=1.5pt] 
{(\getrefnumber{u_1(o_11)} + \getrefnumber{u_2(o_11)})*x - \getrefnumber{u_2(o_11)}};
\addlegendentry{$o_1$}

\addplot[domain=-7:7, color=red, line width=1.5pt, dashed]
{(\getrefnumber{u_1(o_12)} + \getrefnumber{u_2(o_12)})*x - \getrefnumber{u_2(o_12)}};
\addlegendentry{$o_2$}

\addplot[domain=-7:7, color=purple, dotted]
{(\getrefnumber{u_1(o_13)} + \getrefnumber{u_2(o_13)})*x - \getrefnumber{u_2(o_13)}};
\addlegendentry{$o_3$}

\addplot[domain=-7:7, color=green]
{(\getrefnumber{u_1(o_14)} + \getrefnumber{u_2(o_14)})*x - \getrefnumber{u_2(o_14)}};
\addlegendentry{$o_4$}

\addplot[domain=-7:7, color=brown, line width=1.5pt, dotted]
{(\getrefnumber{u_1(o_21)} + \getrefnumber{u_2(o_21)})*x - \getrefnumber{u_2(o_21)}};
\addlegendentry{$o_5$}

\addplot[domain=-7:7, color=orange, dashed]
{(\getrefnumber{u_1(o_22)} + \getrefnumber{u_2(o_22)})*x - \getrefnumber{u_2(o_22)}};
\addlegendentry{$o_6$}

\end{axis}
\end{tikzpicture}
\end{center}
\caption{The corresponding lines for the items in the example.}
\label{fig:lines}
\Description{A coordinate system with $w_1$ as the x-axis and $f_o(w_1)$ as the y-axis, containing 6 lines, one for each chore from the above example.}
\end{figure}
The meaning of each point of intersection is a possible switching point for these two items between the agents.
Clearly, the replacement will only take place between exchangeable pairs, i.e. items in the same category, which are in different agents' bundles at the time of the intersection.
According to Definition \ref{diff_ratio_definition}, at each intersection point of the lines of $o_1$ and $o_2$, $\frac{w_1}{w_2} = \frac{u_2(o_1)-u_2(o_2)}{u_1(o_1)-u_1(o_2)} = r_{2/1}(o_1, o_2)$ holds.  
The largest $r$ value is obtained on the right side of the graph, and as we progress to the left side its value decreases.

In this example, the algorithm starts with the allocation $A=(A_1,A_2)$ in the point $(0.5,0.5)$, which is $A_1=\{o_1,o_2,o_6\}, A_2=\{o_3,o_4,o_5\}$.
Note that for each category, 1's items are the top lines.
In this initial allocation, 2 envies by more than one item, so we start exchanging items in order to increase $w_2$.
The first intersecting pair (when we go left) is $o_5,o_6$. It is an exchangeable pair, so we exchange it and update the allocation to $A_1=\{o_1,o_2,o_5\}, A_2=\{o_3,o_4,o_6\}$. This is an EF1 allocation, so we are done.

If at some point there are multiple intersections of exchangeable pairs, we swap the pairs in an arbitrary order.

\begin{algorithm}[tb]
\caption{Finding an EF[1,1] and PO division for two agents}
\label{alg:main}
\begin{algorithmic}[1]
\Statex{// Step 1: Find a $w$-maximal feasible allocation that is EF for some agent.}
\State \label{state:first-allocation} $A=(A_1,A_2) \gets$ a $w$-maximal allocation, for $w_1 = w_2 = 0.5$.
\If{$A$ is EF[1,1]}
    \State \Return $A$
\EndIf
\If{$A$ is EF for agent 2}
    \State replace the names of agent 1 and agent 2
\EndIf
\Statex{// We can now assume that agent 2 is jealous.}
\Statex{// Step 2: Build a set of item-pairs whose replacement increases agent 2's utility:}
\State \label{state:find-item_pairs} item-pairs $\gets$ all the exchangeable pairs $o_1, o_2\in A_1,A_2$, for which $u_2(o_1)>u_2(o_2)$. 
\State \label{state:find-current-pair} current-pair $\gets$ $(o_1,o_2)$ where $r_{2/1}(o_1,o_2)$ is maximal.
\Statex{// Step 3: Switch items in order until an EF[1,1] allocation is found:} 
\While{$A=(A_1,A_2)$ is not EF[1,1]}
    \State Switch current-pair between the agents.
    \State Update item-pairs list and current-pair (Steps \ref{state:find-item_pairs}, \ref{state:find-current-pair}).
\EndWhile
\State \Return $A$
\end{algorithmic}
\end{algorithm}

\begin{lemma}
\label{lem:first-and-last}
If Algorithm \ref{alg:main} exchanges the last exchangeable pair in the item-pairs list (that is initialized in step \ref{state:find-item_pairs}), then the resulting allocation is envy-free for agent 2.
\end{lemma}
\begin{proof}
After the last exchange, there is no exchangeable pair $(o_1,o_2), o_1,o_2\in A_1,A_2$ for which $o_1$ is the preferred item. Therefore, by Lemma \ref{lem:envy}, agent 2 is not jealous.
\end{proof}

\begin{theorem}
Algorithm \ref{alg:main} always returns an allocation that is $w$-maximal with positive weights (and thus PO), and satisfies the capacity constraints.
The allocation is EF[1,1], and EF1 for a same-sign instance.
\end{theorem}
\begin{proof}
A matching in $G_w$ graph always gives each agent $s_c$ items of category $C_c$. Thanks to the dummy items, all possible allocations that satisfy the capacity constraints can be obtained by a matching.
The first allocation that the algorithm checks is some $w$-maximal allocation, where $w=(w_1,w_2), w_1,w_2\in (0,1)$, so by Proposition \ref{all-PO}, this is a PO allocation.
At each iteration, it exchanges an exchangeable pair, $(o_1,o_2)$, such that $u_2(o_1)>u_2(o_2)$, and among all the exchangeable pairs with $u_2(o_1)>u_2(o_2)$ it has the largest $r_{2/1}(o_1,o_2)$, so by Lemma \ref{lem:exchange}, the resulting allocation is also $w'$-maximal for some $w'=(w_1',w_2'), w_1',w_2'\geq 0$.
In addition, since the items are in the same category, the allocation remains feasible.
The first allocation in the sequence is, by step \ref{state:first-allocation}, envy-free for agent 1.
By Lemma \ref{lem:first-and-last}, the last allocation in the sequence is 
envy-free for agent 2.
So by Lemma \ref{exists-PO-and-EF1}, there exists some iteration in which the allocation is PO and EF[1,1], and EF1 for a same-sign instance.
\end{proof}

\balance

\begin{theorem}
\label{runtime}
    The runtime of Algorithm \ref{alg:main} is $O(m^4)$.
\end{theorem}
\begin{proof}
Step 1 can be done by finding a maximum weighted matching in a bipartite graph $G_w$. Its time complexity is $O(|V|)^3$ (\citet{fredman1987}), where $|V|=2m$, the number of vertices in the graph.
Thus, $O(m^3)$ is the time complexity of step 1.

At step 2 we go through all the categories $c\in [k]$, at each we create groups $A_{1,c},A_{2,c}$ which contain agent 1's and agent 2's items from $C_c$ in $A$.
It can be done in $\frac{m}{2}|C_c|=ms_c$. 
Now we have $|A_{1,c}|=|A_{2,c}|=s_c$.
Then, we iterate over all the pairs $o_1,o_2\in A_{1,c},A_{2,c}$, and add them to the list, which takes $s_c^2$ time.
In total, building item-pairs list is 
$\sum_{c\in [k]}(ms_c+s_c^2) = O(\sum_{c\in [k]}ms_c) = O(km^2)$.
The item-pairs list size is $\sum_{c\in [k]}s_c^2=O(m^2)$, and then finding its maximum takes $O(m^2)$.
In total, step 2 takes $O(km^2)$ time.

The upper bound on the number of iterations in the while loop at step 3 is the number of intersection points between items, which is at most $O(m^2)$.
At each iteration we switch one exchangeable pair, $(o_1,o_2)$, and update the pairs-list. The only pairs that should be updated (deleted or added) are those that contain $o_1$ or $o_2$. There are at most $2 m=O(m)$ such pairs.
Finding the maximum is $O(m^2)$.
In total, step 3 takes $O(m^4)$ time.

Overall, the time complexity of the algorithm is $O(m^4)$ (because $m\geq k$ necessarily).
\end{proof}

\section{Conclusion and Future Work}
We presented the first algorithm for efficient nearly-fair allocation of mixed goods and chores with capacity constraints.
We believe that our paper provides a good first step in understanding fair division of mixed resources under cardinality constraints. 
Our proofs are modular, and some of our lemmas can be used in more general settings.

\subsection{Three or More Agents}
The most interesting challenge is to generalize our algorithm to three or more agents. 
Proposition \ref{all-PO} and Lemmas \ref{both-EF1-for-someone}, \ref{lem:diffs}, \ref{cases-lemma}, \ref{lem:envy}
work for any number of agents, but the other lemmas currently work only for two agents.

Algorithm \ref{alg:main} essentially scans the space of $w$-maximal allocations: it starts with one $w$-maximal allocation, and then moves in the direction that increases the utility of the envious agent.
To extend it to $n$ agents, we can similarly start with a $w$-maximal allocation corresponding to $w = (1/n,\ldots,1/n)$, i.e., identical weights for each of the agents. 
These weights represent a point in an $n$-dimensional space.
Then, we can exchange items to benefit 
an envious agent, in order to increase their weight and improve their utility. In case there are several envious agents, we can select one that is at the “bottom” of the envy chain. For example, in the SWAP algorithm of \citet{biswas2018}, the swap is done in a way that benefits the envious agent with the smallest utility. Similarly, in the envy-graph algorithm of \citet{lipton2004}, the next item is given to an agent with no incoming edges in the envy-graph (an agent who is not envied by any other agent). 
The exchanges should be done in an order that preserves the $w$-maximality and ensures we reach an EF[1,1] allocation.
The two main Lemmas that should be extended to ensure the above two conditions are Lemma \ref{lem:exchange} and Lemma \ref{exists-PO-and-EF1}.
We have not yet been able to develop such a method and prove its correctness.
Finding an EF1+PO allocation for $n=3$ agents seems hard even when there is a single category with only goods.

\subsection{More General Constraints}
Another possible generalization is to more general constraints.
Capacity constraints are a special case of \emph{matroid constraints},
by which each bundle should be an independent set of a given matroid (see  \cite{biswas2018} for the definitions). 
Lemmas \ref{all-PO}, \ref{both-EF1-for-someone},
\ref{exists-PO-and-EF1},
\ref{cases-lemma}
and \ref{lem:exchange} 
do not use categories, and should work for general matroids. The other lemmas should be adapted.

Finally, we assumed that both agents have the same capacity constraints. We do not know if our results can be extended to agents with different capacity constraints (e.g. agent 1 can get at most 7 items while agent 2 can get at most 3 items). Specifically, the proof of Lemma \ref{both-EF1-for-someone} does not work --- if $(A_1,A_2)$ is feasible, then $(A_2,A_1)$ might be infeasible.

\newpage 



\begin{acks}
This research has been partly supported by the Ministry of Science, Technology \& Space (MOST), Israel.
\end{acks}



\bibliographystyle{ACM-Reference-Format} 
\bibliography{sample}

\clearpage
\appendix
\section*{Appendix}
\section{Methods that Do Not Work}
\label{app:dont-work}
In this section, we present some of our attempts to find an EF1 and PO allocation for an instance with chores and capacity constraints, using ideas from previous works.
These attempts failed. This shows that the problem is not trivial, and the new tools that have been developed in this paper are required.

\subsection{Iterated Matching Algorithm}
The Iterated Matching algorithm, presented by \citet{brustle2019}, finds an EF1 allocation of indivisible goods in the case of additive valuations.
This is done by using the \emph{valuation graph}, which is a complete bipartite graph on vertex sets $I$ ($n$ agents) and $J$ ($m$ goods), with weights that represent the agents' utilities.
The algorithm proceeds in rounds where each agent is matched to exactly one item in each round, by a maximum weighted matching on a sub graph of all the remaining goods, until all items have been allocated.

This algorithm can be easily applied to chores with one category, by adding at the beginning some $k$ dummy chores (with utility of 0 to each agent), where $|J|=an-k$, and $a\in \mathds{N}, k\in \{0,...,n-1\}$. 

However, with more than one category, we need to add an external loop that runs on all categories, and at each iteration executes the algorithm for chores.
While it maintains capacity constraints (because in each iteration all the agents get chores, similarly to round robin procedure), it may not necessarily maintain the EF1 requirement, as we show in the following example. 

Denote by $o_{i,j}$ the $j$-th item of category $i$. 
Table \ref{tab:iterated-matching} presents the utilities of the agents over the items.

\begin{table}[]
    \caption{Iterated matching algorithm counterexample}
    \label{tab:iterated-matching}
    \centering
    \begin{tabular}{||c| c c | c c||} 
    \hline
    {} & {$o_{1,1}$} & {$o_{1,2}$} & {$o_{2,1}$} & {$o_{2,2}$} \\ [0.5ex] 
    \hline\hline
    Agent 1 & 0 & -2 & -2 & -1 \\
    \hline
    Agent 2 & 0 & -4 & -4 & 0 \\
    \hline
    \end{tabular} 
\end{table}

After allocating category $C_1$, the allocation is $A_1=\{o_{1,2}\}, A_2=\{o_{1,1}\}$, so $u_1(A_1)=-2, u_1(A_2)=0, u_2(A_1)=-4, u_2(A_2)=0$, then agent 1 envies 2 up to one item (her only item), and agent 2 is not jealous.

Then we allocate the second category, which changes the allocations to: $A_1=\{o_{1,2},o_{2,1}\}, A_2=\{o_{1,1},o_{2,2}\}$. Now $u_1(A_1)=-4, u_1(A_2)=-1$, so agent 1 envies by more than one item (her worst chore has utility of $-2$).

If there was always an agent who was not jealous, we could have assigned her the new chore, but this is not guaranteed. An envy-cycle may be created, and we know that envy-cycle elimination algorithm may fail EF1 for additive chores [according to \citet{vaish2020}].

\subsection{Top-trading Envy Cycle Elimination Algorithm}
\citet{vaish2020} considered  fair allocation of chores, and suggested to use cycle elimination on the \emph{top-trading graph}, instead of the usual envy-graph.
The top-trading graph for a division $A$ is a directed graph on the vertex set $N$, with a directed edge from agent $i$ to agent $k$ if $u_i(A_k) = \max_{j \in N} u_i(A_j)$ and $u_i(A_k) > u_i(A_i)$, i.e. $A_k$ is the most preferred bundle for agent $i$ in $A$, and she strictly prefers $A_k$ over her own bundle.
In their paper \cite{vaish2020}, they show that resolving a top-trading envy cycle preserves EF1. Indeed, every agent involved in the top-trading exchange receives its most preferred bundle after the swap, and therefore does not envy anyone else in the next round.
They also define a \emph{sink agent} as an agent with no out-going edges in the envy graph,
that is, an agent who does not envy anybody.
In addition, they prove that if the usual envy-graph does not have a sink, then the top-trading envy graph has a cycle. 

In their algorithm, for each chore, they construct the envy-graph. If there is no sink in it, they eliminate cycles on the top-trading envy-graph, which guarantees the existence of a sink agent in the envy graph, and then allocate the chore to a sink agent.

This method does not work in the setting with capacity constraints, because we can not simply assign the new chore to the sink agent, because she may have reached the maximum allowed number of chores from this category.

For example, consider an instance with two categories, $C_1=\{o_{1,1}, o_{1,2}, o_{1,3}, o_{1,4}\}$ and $C_2=\{o_{2,1}, o_{2,2}\}$, with capacity constraints $S=\{2,1\}$, and utility functions presented in Table \ref{tab:top-trading}.

\begin{table}[]
    \caption{Top trading algorithm counterexample}
    \label{tab:top-trading}
    \centering
    \begin{tabular}{||c | c c c c | c c||} 
    \hline
    {} & {$o_{1,1}$} & {$o_{1,2}$} & {$o_{1,3}$} & {$o_{1,4}$} & {$o_{2,1}$} & {$o_{2,2}$} \\ [0.5ex] 
    \hline\hline
    Agent 1 & -1 & 0 & 0 & 0 & -2 & -4 \\
    \hline
    Agent 2 & -1 & 0 & 0 & 0 & -1 & -3 \\
    \hline
    \end{tabular} 
\end{table}
At the beginning of the algorithm, both allocations are empty, so agent 1 and agent 2 are both sinks. 
Say that agent 1 was selected to get $o_{1,1}$, and now $A_1=\{o_{1,1}\}, A_2=\emptyset$.
Now agent 1 is jealous, so the only sink agent is 2, and $o_{1,2}$ is allocated to 2.
Since $u_2(o_{1,2})=u_2(o_{1,3})=0$, agent 2 remains sink in the two following iterations. Then, the new allocations are $A_1=\{o_{1,1}\}, A_2=\{o_{1,2},o_{1,3}\}$, and the only sink is agent 2. 
According to the algorithm, we should assign $o_{1,4}$ to agent 2, but this violates the capacity constraints. 

\subsection{Greedy Round-robin with Cycle Elimination}

\citet{biswas2018} solved the problem of allocating goods under capacity constraints. 
Their algorithm first determines an arbitrary ordering of the n agents, $\sigma$, and then for each category: uses the Greedy Round-Robin algorithm to allocate the goods of this category, eliminate the cycles on the envy-graph, and update $\sigma$ to be a topological ordering of the envy-graph. 

As already mentioned, this algorithm will not work for chores, because eliminating cycles in the usual envy-graph may violate EF1 \cite{vaish2020}.

In addition, if we use the top-trading graph instead, 
we can not use it to determine the topological ordering, since this ordering should be based on the envy-graph.
It is possible that the top-trading graph is cycle-free, while the envy-graph has cycles.

For example, consider an instance with 4 agents, one category with 4 items, $C_1=\{o_1,o_2,o_3,o_4\}$, a capacity constraint of 1, and the utilities presented in Table \ref{tab:CycleElimination}. 

\begin{table}[]
    \caption{Cycle elimination algorithm counterexample}
    \label{tab:CycleElimination}
    \centering
    \begin{tabular}{||c | c c c c ||} 
    \hline
    {} & {$o_1$} & {$o_2$} & {$o_3$} & {$o_4$} \\ [0.5ex] 
    \hline\hline
    Agent 1 & -5 & -3 & -7 & -7 \\
    \hline
    Agent 2 & -5 & -2 & -1 & -4 \\ 
    \hline
    Agent 3 & -4 & -7 & -6 & -1 \\ 
    \hline
    Agent 4 & -3 & -3 & -2 & -1 \\ 
    \hline
    \end{tabular} 
\end{table}

Consider the allocation $A=(A_1,A_2,A_3,A_4)$, where $A_1=\{o_1\},A_2=\{o_2\},A_3=\{o_3\},A_4=\{o_4\}$. 
The envy graph of allocation $A$ is  shown in Figure \ref{fig:envy-graph}, and its top-trading graph is shown in Figure \ref{fig:top-trading-graph}.

\begin{figure}[!htb]
\begin{minipage}{0.20\textwidth}
    \centering
    \caption{Envy-graph}
    \label{fig:envy-graph}
    \begin{tikzpicture}[node distance={15mm}, main/.style = {draw, circle}]
    \node[main] (1) {$1$}; 
    \node[main] (2) [above right of=1] {$2$}; 
    \node[main] (3) [below right of=1] {$3$}; 
    \node[main] (4) [above right of=3] {$4$}; 
    \draw[->] (1) -- (2);
    \draw[->] (2) -- (3);
    \draw[->] (3) -- (1);
    \draw[->] (3) -- (4);
    \end{tikzpicture} 
\end{minipage}\hfill
\begin{minipage}{0.22\textwidth}
    \centering
    \caption{Top-trading graph}
    \label{fig:top-trading-graph}
    \begin{tikzpicture}[node distance={15mm}, main/.style = {draw, circle}]
    \node[main] (1) {$1$}; 
    \node[main] (2) [above right of=1] {$2$}; 
    \node[main] (3) [below right of=1] {$3$}; 
    \node[main] (4) [above right of=3] {$4$}; 
    \draw[->] (1) -- (2);
    \draw[->] (2) -- (3);
    \draw[->] (3) -- (4);
    \end{tikzpicture}
   \end{minipage}
\end{figure}

Note that the top-trading graph is cycle-free, while the envy-graph has a cycle ($1\rightarrow 2 \rightarrow 3$), So we do not have a topological ordering on it. 

\subsection{Pareto-improve an EF1 Allocation}
We examined the approach of finding an EF1 allocation, which is not necessarily PO, and applying Pareto improvements to it until an EF1 and PO allocation is obtained. However, the following proposition shows that this approach is inadequate, even with two agents.

\begin{proposition}
Not every Pareto-improvement of an EF1 allocation yields an EF1 allocation, even with two agents. 
\end{proposition} 

\begin{proof}
Let $A = \{A_1,A_2\}$ be an EF1 allocation for two agents.
Suppose that all the items are chores.
Define a Pareto improvement as a replacement between two subsets of chores: $X_1 \subseteq A_1$ and $X_2 \subseteq A_2$ (one of the subsets may be empty), such that the change harms no one and benefits at least one agent.
In particular, $\forall i \in {1,2}$ : $u_i(X_{3-i}) \geq u_i(X_i)$ --- the agent prefers (or indifferent) what she received over what she gave.

The following example proves the proposition. Consider
an instance with one category with 8 chores, capacity constraint of 8, and two agents with the valuations presented in Table \ref{tab:pareto-improve}.

\begin{table}[]
    \caption{Pareto improvements counterexample}
    \label{tab:pareto-improve}
    \centering
    \begin{tabular}{||c | c c c c c c c c||} 
    \hline
    {} & {$o_1$} & {$o_2$} & {$o_3$} & {$o_4$} & {$o_5$} & {$o_6$} & {$o_7$} & {$o_8$} \\ [0.5ex] 
    \hline\hline
    Agent 1 & -5 & -2 & -1 & -2 & -2 & -2 & -1 & -2 \\
    \hline
    Agent 2 & -1 & -1 & -2 & -1 & -1 & 0 & 0 & 0 \\ 
    \hline
    \end{tabular} 
\end{table}
Suppose that the EF1 allocation, $A$, is $A_1=\{o_1,o_5,o_6,o_7\}$, $A_2=\{o_2,o_3,o_4,o_8\}$. 
The utilities of the agents in $A$ are: 
\begin{itemize}
\item $u_1(A_1) = -10, u_1(A_2) = -7$
\item $u_2(A_1) = -2, u_2(A_2) = -4$
\end{itemize}

Clearly, the two agents are jealous of each other, but their envy is up to one chore because $u_1(A_1 \diagdown \{o_1\}) \geq u_1(A_2)$, and $u_2(A_2 \diagdown \{o_3\}) \geq u_2(A_1)$. 

In addition, $A$ is not PO because there is an envy-cycle in the envy-graph.

Consider the following Pareto-improvement: $X_1=\{o_1\}, X_2=\{o_2,o_3,o_4\}$. It  does not harm agent 1 and benefits agent 2. 

After the replacement, the utilities of agent 1 do not change, that is, $u_1(A_1) = -10, u_1(A_2) = -7$. However, the most difficult chore in agent 1's bundle is worth $-2$, which is not enough for her to eliminate the envy. So the Pareto-improvement is not EF1.
\end{proof}

\section{Technical Lemmas}

\begin{lemma}
\label{lem:diffs_mul}
For any six real numbers $x_i,x_j,y_i,y_j,z_i,z_j$, the following inequalities are equivalent:
\begin{align}
\label{eq:basic}
x_i(y_j-z_j) + y_i(z_j-x_j) + z_i(x_j-y_j) \leq 0
\\
\label{eq:x}
(x_i-y_i)(x_j-z_j)
\leq
(x_i-z_i)(x_j-y_j)
\\
\label{eq:y}
(y_i-z_i)(y_j-x_j)
\leq
(y_i-x_i)(y_j-z_j)
\\
\label{eq:z}
(z_i-x_i)(z_j-y_j)
\leq
(z_i-y_i)(z_j-x_j)
\end{align}
\end{lemma}
\begin{proof}
By adding $+x_i x_j - x_i x_j$ to  inequality \eqref{eq:basic}, we get:
\begin{align*}
&
\eqref{eq:basic}
\\
\iff&
x_i(y_j-z_j+x_j-x_j) + y_i(z_j-x_j) + z_i(x_j-y_j) \leq 0
\\
\iff&
x_i((x_j-z_j)-(x_j-y_j)) 
-
y_i(x_j-z_j) 
+
z_i(x_j-y_j) 
\leq 0
\\
\iff&
(x_i-y_i)(x_j-z_j)
-
(x_i-z_i)(x_j-y_j)
\leq 0
\\
\iff&
(x_i-y_i)(x_j-z_j)
\leq
(x_i-z_i)(x_j-y_j)
\\
\iff& \eqref{eq:x}.
\end{align*}

Inequalities \eqref{eq:y}, \eqref{eq:z} are entirely analogous.
\end{proof}

\begin{lemma}
\label{lem:diffs_div}
If $x_i > y_i > z_i$, then the following are equivalent:

\begin{align}
\frac{x_j-z_j}{x_i-z_i}  
\leq
\frac{x_j-y_j}{x_i-y_i}  
&& r_{j/i}(x,z) \leq r_{j/i}(x,y)
\\
\frac{y_j-z_j}{y_i-z_i}
\leq
\frac{y_j-x_j}{y_i-x_i}
&& r_{j/i}(y,z) \leq r_{j/i}(x,y)
\\
\frac{z_j-y_j}{z_i-y_i}
\leq
\frac{z_j-x_j}{z_i-x_i}
&& r_{j/i}(y,z) \leq r_{j/i}(x,z)
\end{align}
\end{lemma}
\begin{proof}
For each $k$ in \{\ref{eq:x},\ref{eq:y},\ref{eq:z}\},
divide inequality $(k)$ by the two terms with the ``$_i$'' subscript to get inequality $(k+3)$.
\end{proof}

\begin{observation}
\label{obs}
Lemma \ref{lem:diffs_mul} is still true if we reverse all inequalities directions, so too Lemma \ref{lem:diffs_div}.

That is,
if $x_i > y_i > z_i$, 
then the following are equivalent:

\begin{align}
\frac{x_j-z_j}{x_i-z_i}  
\geq
\frac{x_j-y_j}{x_i-y_i}  
&& r_{j/i}(x,z) \geq r_{j/i}(x,y)
\\
\frac{y_j-z_j}{y_i-z_i}
\geq
\frac{y_j-x_j}{y_i-x_i}
&& r_{j/i}(y,z) \geq r_{j/i}(x,y)
\\
\frac{z_j-y_j}{z_i-y_i}
\geq
\frac{z_j-x_j}{z_i-x_i}
&& r_{j/i}(y,z) \geq r_{j/i}(x,z)
\end{align}
\end{observation}

\section{A Complete Proof for Lemma \ref{lem:diffs}}
\label{app:lem:diffs}

\begin{algorithm}[tb]
\caption{Transforming one allocation to another one using edge-disjoint exchange-cycles}
\label{alg:find-exchange-cycles}
\textbf{Input:} Two feasible allocations $A=(A_1,...,A_n), A'=(A_1',...,A_n')$ \\
\hspace*{\algorithmicindent} \textbf{Output:} A sequence of exchange-cycles leading from $A$ to $A'$.
\begin{algorithmic}[1]
\State \label{state:init}result $\gets$ an empty sequence.
\State \label{state:first-if}$G\gets$ an empty graph, in which the nodes are the $n$ agents.
\State Choose an item $o_i$ such that $o_i\in A_i, o_i\in A_j'$, for some $i,j\in [n]$
\If {there is no such item}
\State \Return result
\EndIf
\State Add a directed edge $(i,j)$ to $G$ with name $o_i$.
\Statex // Since $j$ gets a new item $(o_i)$, and in every feasible allocation there is the same number of items from each category, $j$ must give someone an item from the same category.
\State \label{state:first-else}$o_j\gets$ the item from $o_i$'s category such that $o_j\in A_j, o_j\in A_k'$, for some $k\in [n]$
\State Add a directed edge $(j,k)$ to $G$ with name $o_j$.
\If {$G$ has a cycle ($k==i$)}
\State Append the cycle to result.
\State $A \gets$ the allocation after exchanging items in $A$ according to this cycle.
\State Go to \ref{state:first-if}
\Else
\State $j\gets k$
\State Go to \ref{state:first-else}
\EndIf
\end{algorithmic}
\end{algorithm}

\begin{customthm}{\ref{lem:diffs}}
For any $n$ agents,
for any $w=(w_1,w_2,...,w_n)$ such that $w_1,w_2,...,w_n\in(0,1)$, and an allocation $A=(A_1,...,A_n)$, the following are equivalent:

(i) $A$ is $w$-maximal.


(ii) Every exchange-cycle does not increase the weighted sum of utilities.
That is, for all $x\geq 2$, a subset of agents $\{a_1,...,a_x\}\in [n]$, and a set of items $o_1,...,o_x$, such that all are in the same category, and $\forall j\in [x], o_j\in A_{a_j}$:
\begin{align*}
    w_{a_1} u_{a_1} (o_1) + w_{a_2} u_{a_2} (o_2) + ... + w_{a_x} u_{a_x} (o_x) \geq \\
    w_{a_1} u_{a_1} (o_x) + w_{a_2} u_{a_2} (o_1) + ... + w_{a_x} u_{a_x} (o_{x-1})
\end{align*}
\end{customthm}
\begin{proof}~

(i) $\implies$ (ii):
Suppose $A$ is a $w$-maximal allocation. Suppose toward a contradiction that (ii) is not true, that is, there exists a set of indices $\{a_1,...,a_x\}\in [n]$ and a set of items in the same category $o_1,...,o_x$, such that:
\begin{align*}
    w_{a_1} u_{a_1} (o_1) + w_{a_2} u_{a_2} (o_2) + ... + w_{a_x} u_{a_x} (o_x) < \\
    w_{a_1} u_{a_1} (o_x) + w_{a_2} u_{a_2} (o_1) + ... + w_{a_x} u_{a_x} (o_{x-1})
\end{align*}
So we can switch those items in a cycle, and get another feasible allocation with a weighted sum greater than $A$. Thus $A$ is not $w$-maximal. Contradiction.

(ii) $\implies$ (i): 
Denote by $A'=(A_1',...,A_n')$ some feasible allocation, different from $A$.
We claim that it is possible to transform $A$ to $A'$ using a sequence of pairwise-disjoint exchange-cycles.
One way to find these exchange-cycles is presented in Algorithm \ref{alg:find-exchange-cycles}. 
Note that, each cycle in the sequence places the items involved in the exchange exactly where they should be according to the $A'$ allocation. Therefore, any item involved in one cycle, cannot be involved in any other cycle, that is, the cycles have no items in common.
Therefore, all cycles found by Algorithm \ref{alg:find-exchange-cycles} exist in allocation $A$, so according to the assumption in (ii), each of them does not increase the weighted sum of the allocation.
Therefore, $w_1u_1(A_1')+...+w_nu_n(A_n') \leq w_1u_1(A_1)+...+w_nu_n(A_n)$.

This holds for every feasible allocation $A'$, so $A$ is $w$-maximal by definition. 
\end{proof}

\section{A Complete Proof for Lemma \ref{lem:exchange}}
\label{app:lem:exchange}

\begin{customthm}{\ref{lem:exchange}}
Suppose there are $n=2$ agents.
Let $A$ be a $w$-maximal allocation, for $w=(w_1,w_2)$. Suppose there is an exchangeable pair $o_1\in A_1, o_2\in A_2$ such that:
\begin{enumerate}
\item $u_2(o_1)>u_2(o_2)$, that is, $o_1$ is the preferred item.
\item Among all exchangeable pairs with a preferred item, this pair has the largest difference-ratio $r_{2/1}(o_1,o_2)$.
\end{enumerate}
Let $A'$ be the allocation resulting from exchanging $o_1$ and $o_2$ in $A$.
Then, $A'$ is $w'$-maximal for some $w' = (w_1',w_2')$ with  $w_1'\leq w_1, w_2'\geq w_2, w_1'\in(0,1), w_2'\in(0,1)$.
\end{customthm}
\begin{proof}
By Lemma \ref{cases-lemma}, $u_1(o_1) > u_1(o_2)$, because $A$ is $w$-maximal and $(o_1,o_2)$ is an exchangeable pair for which $u_2(o_1)>u_2(o_2)$.
Then, by definition, $r_{2/1}(o_1,o_2) > 0$, and by lemma \ref{lem:diffs2}(ii), $\frac{w_1}{w_2} \geq r_{2/1}(o_1,o_2)$.

Consider the allocation $A'=(A_1',A_2')$, resulting from exchanging $o_1$ and $o_2$. Consider some $w_1',w_2'\in (0,1)$ such that $\frac{w_1'}{w_2'} = r_{2/1}(o_1,o_2)$, so $0 < \frac{w_1'}{w_2'} \leq \frac{w_1}{w_2}$.

We now look at all the exchangeable pairs $(o_1^*,o_2^*)$ after the exchange, and see that they satisfy all the conditions of Lemma \ref{lem:diffs2}(ii) with $w'=(w_1',w_2')$, which can be written as:
\begin{enumerate}[(a)] 
\item $u_1(o_1^*) > u_1(o_2^*)$ and 
    $r_{2/1}(o_1,o_2) \geq r_{2/1}(o_1^*,o_2^*)$ or
\item $u_1(o_1^*) = u_1(o_2^*)$ and
    $u_2(o_2^*)\geq u_2(o_1^*)$ or
\item $u_1(o_1) < u_1(o_2)$ and
    $r_{2/1}(o_1,o_2) \leq r_{2/1}(o_1^*,o_2^*)$
\end{enumerate}

\begin{enumerate}
    \item The exchangeable pairs $(o_1^*,o_2^*)$ who have not moved: 
Lemma \ref{lem:diffs2} implies that:
\begin{itemize}
\item If $u_1(o_1^*) > u_1(o_2^*)$, then $\frac{w_1}{w_2} \geq r_{2/1}(o_1^*,o_2^*)$. If also $u_2(o_1^*) > u_2(o_2^*)$, then the maximality of $r_{2/1}(o_1,o_2)$ [condition \ref{max-cond}] implies $\frac{w_1}{w_2} \geq r_{2/1}(o_1,o_2) \geq r_{2/1}(o_1^*,o_2^*)$. 
Else, by Definition \ref{diff_ratio_definition}, $r_{2/1}(o_1^*,o_2^*) \leq 0$, so $r_{2/1}(o_1,o_2) \geq 0 \geq r_{2/1}(o_1^*,o_2^*)$ holds again.
\item If $u_1(o_1^*) = u_1(o_2^*)$, then $u_2(o_2^*) \geq u_2(o_1^*)$.
\item If $u_1(o_1^*) < u_1(o_2^*)$, then $\frac{w_1}{w_2} \leq r_{2/1}(o_1^*,o_2^*)$. In particular $r_{2/1}(o_1,o_2) \leq r_{2/1}(o_1^*,o_2^*)$.
\end{itemize} 
After the exchange, they still satisfy the same conditions.

    \item The pair $(o_2,o_1)$:

    Now $o_2\in A_1', o_1\in A_2'$, and the pair $(o_2,o_1)$ fits condition (c), which says that the item in 1's bundle worth less, for agent 1, than the item in 2's bundle ($u_1(o_2) < u_1(o_1)$).
    For it $\frac{w_1'}{w_2'} = r_{2/1}(o_1,o_2) = r_{2/1}(o_2,o_1)$, by definition and by symmetry of $r$.

    \item Pairs in the form $(o_1^*,o_1)$, $o_1^*\in A_1', o_1^* \neq o_2$:
    \begin{enumerate}[(a)] 
    
    \item If $u_1(o_1^*)>u_1(o_1)$, of course $u_1(o_1^*)>u_1(o_1)>u_1(o_2)$.
        If also $u_2(o_1^*)>u_2(o_1)$ then because of the maximality condition, 
        $r_{2/1}(o_1,o_2)\geq r_{2/1}(o_1^*,o_2)$.
        By \ref{obs} (with $x=o_1^*, y=o_1, z=o_2$), this is equivalent to 
        $r_{2/1}(o_1,o_2)\geq r_{2/1}(o_1^*,o_1)$.
        And if $u_2(o_1^*)\leq u_2(o_1)$, then $r_{2/1}(o_1^*,o_1) \leq 0 < r_{2/1}(o_1,o_2)$.
    
\item If $u_1(o_1^*)=u_1(o_1)$,
it is not possible that $u_2(o_1)<u_2(o_1^*)$ because it implies $u_2(o_2)<u_2(o_1^*)$, and by \ref{cases-lemma}, $u_1(o_2)<u_1(o_1^*)$.
By the values of r's numerators and denominators, we get $r_{2/1}(o_1^*,o_2) > r_{2/1}(o_1,o_2)$, and it contradicts $r_{2/1}(o_1,o_2)$ maximality. 
Therefore, $u_2(o_1)\geq u_2(o_1^*)$.
    
\item If $u_1(o_1^*)<u_1(o_1)$,
it is also not possible that $u_2(o_1) < u_2(o_1^*)$, as explained in (b). It is also not 
possible that $u_2(o_1) = u_2(o_1^*)$,
because then $u_2(o_1^*) > u_2(o_2)$, and by \ref{cases-lemma}, since $A$ is an $(w_1,w_2)$-maximal and $(o_1^*,o_2)$ is an exchangeable pair in it, $u_1(o_1^*)>u_1(o_2)$. Therefore, $r_{2/1}(o_1^*,o_2) > r_{2/1}(o_1,o_2)$, contradiction. 
    
So in that case, necessarily $u_2(o_1^*)<u_2(o_1)$. 
        Based on that, we now show that $r_{2/1}(o_1,o_2) \leq r_{2/1}(o_1^*,o_1)$.
    
        If also $u_2(o_1^*) > u_2(o_2)$, by $r_{2/1}(o_1,o_2)$ maximality, $r_{2/1}(o_1^*,o_2) \leq r_{2/1}(o_1,o_2)$, 
        and by \ref{cases-lemma}, $u_1(o_1) > u_1(o_1^*) > u_1(o_2)$.
        Then, by \ref{lem:diffs_div} (with $x=o_1,y=o_1^*,z=o_2$), $r_{2/1}(o_1,o_2) \leq r_{2/1}(o_1^*,o_1)$.
        
        Else, $u_2(o_1^*) \leq u_2(o_2)$. Since $u_1(o_2)<u_1(o_1)$ and $u_1(o_1^*)<u_1(o_1)$, there are two options:
\begin{itemize}
\item $u_1(o_1^*)<u_1(o_2)<u_1(o_1)$.

By Lemma \ref{lem:diffs2} we know that $r_{2/1}(o_1,o_2) \leq \frac{w_1}{w_2}, r_{2/1}(o_1^*,o_2) \geq \frac{w_1}{w_2}$, so $r_{2/1}(o_1,o_2) \leq r_{2/1}(o_1^*,o_2)$, and by \ref{obs} (with $x=o_1,y=o_2,z=o_1^*$), $r_{2/1}(o_1,o_2) \leq r_{2/1}(o_1^*,o_1)$.
\item 
$u_1(o_2)\leq u_1(o_1^*)<u_1(o_1)$.

Then, 
$u_2(o_1)-u_2(o_2) \leq u_2(o_1)-u_2(o_1^*)$
and 
$u_1(o_1)-u_1(o_2) \geq u_1(o_1)-u_1(o_1^*)$.
So
$r_{2/1}(o_1,o_2)
= 
\frac{u_2(o_1)-u_2(o_2)}{u_1(o_1)-u_1(o_2)}
\leq 
\frac{u_2(o_1)-u_2(o_1^*)}{u_1(o_1)-u_1(o_1^*)}
=
r_{2/1}(o_1^*,o_1)$.
\end{itemize}
\end{enumerate}
    
    \item Pairs in the form $(o_2,o_2^*)$, $o_2^*\in A_2', o_2^* \neq o_1$:
    
    Exactly the same arguments in case 3 (but with $o_2,o_2^*$), prove this case.
\end{enumerate}
Therefore, by Lemma \ref{lem:diffs2} (ii), $A'$ is $(w_1',w_2')$-maximal, and  $\frac{w_1'}{w_2'} \leq \frac{w_1}{w_2}$ with $w_1'+w_2'=1$ implies $w_1'\leq w_1$ and $w_2' \geq w_2$.
\end{proof}

\end{document}